\documentclass{sl}     

\usepackage{slsec}     
\usepackage{stud_log} 
\usepackage{slfoot}
\usepackage{slthm}  
\usepackage{amssymb,amsfonts,latexsym,stmaryrd}
\usepackage[PostScript=dvips]{diagrams}
\usepackage{url}

\title{Relational Hidden Variables and Non-Locality}
\author{Samson Abramsky\\
Oxford University Computing Laboratory}
\date{}

\newtheorem{theorem}{Theorem}[section]
\newtheorem{lemma}[theorem]{Lemma}

\newtheorem{proposition}[theorem]{Proposition}

\newtheorem{question}[theorem]{Question}
\newtheorem{corollary}[theorem]{Corollary}

\newcommand{\beqa}{\begin{eqnarray*}}
\newcommand{\eeqa}{\end{eqnarray*}\par\noindent}

\newcommand{\rarr}{\rightarrow}
\newcommand{\HH}{\mathcal{H}}

\newcommand{\ket}[1]{{|} #1\rangle}
\newcommand{\bra}[1]{\langle #1{|}}
\newcommand{\ie}{\textit{i.e.}~}

\newcommand{\EE}{\mathbf{E}}

\newcommand{\Comp}{\mathbb{C}}
\newcommand{\Real}{\mathbb{R}}

\newcommand{\dom}{\mathsf{dom}}

\newcommand{\Pow}{\mathcal{P}}

\newcommand{\mb}{\overline{m}}
\newcommand{\ob}{\overline{o}}
\newcommand{\IFF}{\; \leftrightarrow \;}
\newcommand{\siff}{\leftrightarrow}
\newcommand{\AND}{\; \wedge \;}
\newcommand{\OR}{\; \vee \;}
\newcommand{\IMP}{\; \rightarrow \;}
\newcommand{\nn}{\mathbf{n}}
\newcommand{\Mim}{M^{-}_{i}}
\newcommand{\Oim}{O^{-}_{i}}
\newcommand{\da}{{\downarrow}}
\newcommand{\vn}{\varnothing}
\newcommand{\Loc}{\mbox{\textbf{L}}}
\newcommand{\WD}{\textbf{WD}}
\newcommand{\SD}{\textbf{SD}}
\newcommand{\PI}{\textbf{PI}}
\newcommand{\OI}{\textbf{OI}}
\newcommand{\LI}{$\mathbf{\lambda}$\textbf{I}}
\newcommand{\NS}{\textbf{NS}}
\newcommand{\SV}{\textbf{SV}}
\newcommand{\EM}{\mathsf{EM}}
\newcommand{\QM}{\mathsf{QM}}
\newcommand{\NSig}{\mathsf{NS}}
\newcommand{\NSp}{\mathsf{NS}^{p}}
\newcommand{\LHV}{\mathsf{LHV}}
\newcommand{\HV}{\mathsf{HV}}
\newcommand{\NP}{\mathsf{NP}}
\newcommand{\PSPACE}{\mathsf{PSPACE}}
\newcommand{\PLI}{\textbf{P$\mathbf{\lambda}$\textbf{I}}}
\newcommand{\PLoc}{\textbf{PL}}
\newcommand{\PNS}{\textbf{PNS}}
\newcommand{\POI}{\textbf{POI}}
\newcommand{\PPI}{\textbf{PPI}}
\newcommand{\PML}{\textbf{PML}}
\newcommand{\ML}{\textbf{ML}}
\newcommand{\pfn}{\rightharpoonup}
\newcommand{\inst}[2]{\begin{array}{c} #1 \\ #2 \end{array}}

\newcommand{\Klm}{K_{\mb, \lambda}}
\newcommand{\Klim}{K_{\mb, \lambda}^{i}}

\newcommand{\Wlm}{W_{\mb, \lambda}}
\newcommand{\card}[1]{\mathsf{card}(#1)}
\newcommand{\Lp}{\Lambda^{+}}
\newcommand{\Lv}{\Lambda'}
\newcommand{\Mp}{M^{+}}
\newcommand{\Mip}{M^{+}_{i}}
\newcommand{\Oml}{O_{\mb, \lambda}}
\newcommand{\Oiml}{O_{m, \lambda}^{i}}
\newcommand{\tp}{\theta_{p}}
\newcommand{\tq}{\theta_{q}}

\newcommand{\thq}{\theta_{q^{h}}}
\newcommand{\Amo}{A_{\mb, \ob}}
\newarrow{Eq}=====

\begin{document}


\setcounter{page}{1}     

 

\AuthorTitle{Samson Abramsky}{Relational Hidden Variables And Non-Locality}

\PresentedReceived{Jouko V\"a\"an\"anen}{November 2nd, 2010}

\begin{abstract}
We use a simple relational framework to develop the key notions and results on \emph{hidden variables} and \emph{non-locality}. The extensive literature on these topics in the foundations of quantum mechanics is couched in terms of probabilistic models, and properties such as locality and no-signalling are formulated probabilistically.
We show that to a remarkable extent, the main structure of the theory, through the major No-Go theorems and beyond, survives intact under the replacement of probability distributions by mere relations. 
\end{abstract}


\section{Introduction}

In this paper, we consider a simple relational setting, in which the key notions and results concerning \emph{hidden variables} and \emph{non-locality} can be studied. The extensive literature on these topics in the foundations of quantum mechanics is couched in terms of probabilistic models, and properties such as locality and no-signalling are formulated probabilistically.
We show that to a remarkable extent, the main structure of the theory, through the major No-Go theorems and beyond, survives intact under the replacement of probability distributions by mere relations.

The main contents of the paper can be summarized as follows:

\begin{itemize}
\item In the first part of the paper, sections 2--6,
we define purely relational analogues of all the key notions around locality and Bell's theorem which have been formulated in terms of probabilistic models, including Weak and Strong Determinism, No-Signalling, $\lambda$-Independence \cite{dickson1999quantum}, Parameter Independence and Outcome Independence \cite{jarrett1984physical,shimony1986events}, and Locality \cite{bell1964einstein}.
We show that these relational notions have the same logical relationships as their probabilistic counterparts. We give clean proofs, from explicit assumptions, of No-Go theorems based on the EPR \cite{einstein1935can}, GHZ \cite{greenberger1990bell}, Hardy \cite{hardy1993nonlocality}, and Kochen-Specker \cite{kochen1975problem} constructions.

This part of the paper can in large part be seen as a recasting of a recent paper by Brandenburger and Yanofsky \cite{brandenburger2008classification}  in relational form. Their paper, which is couched entirely in terms of probabilistic models, gives a careful, unified  treatment of the major properties of hidden-variable models, and a classification in terms of these. We show that  their results can be replicated in the purely relational setting. We also prove relational analogues of some additional results obtained in a subsequent paper by Brandenburger and Keisler \cite{BK}.

\item Our proof of the GHZ result is based on Mermin's well-known proof \cite{mermin1990quantum}, but by making the assumptions explicit, we find a lacuna in his argument. His `instructions' are deterministic; but determinism is not being assumed \textit{a priori}. We show how determinism can be \emph{derived} from apparently weaker assumptions, by virtue of a general result on hidden-variable models.

\item Our proof of the Hardy result shows that, despite the need for probabilities in the quantum realization of the construction, the No-Go result itself can be proved in purely logical terms. 
It also shows that such results apply in the relational setting even in the bipartite case.

\item Our analysis of the Kochen-Specker result explains a surprising  formal connection between No-Signalling and Contextuality, clarifying an apparent anomaly in \cite{brandenburger2008classification}.

\item We also define an explicit mapping from a class of quantum systems to relational models, and show that the constructions used in the relational no-go theorems are in the image of this mapping, thus obtaining the usual applications of the No-Go theorems to quantum mechanics. This clean separation of the usual arguments into `logical' and `physical' components is in our opinion an attractive feature of this approach.

\item We then go on to consider the connections between probabilistic and relational models. Probabilistic models can be reduced to relational ones by the `possibilistic collapse', in which non-zero probabilities are conflated to (possible) truth.
We show that all the independence properties we have been studying are preserved by the possibilistic collapse, in the sense that if the property in its probabilistic form is satisfied by the probabilistic model, then the relational version of the property will be satisfied by its possibilistic collapse. More surprisingly, we also show a \emph{lifting property}: if a relational model can be realized by local hidden variables, then there is a probabilistic model whose possibilistic collapse gives rise to the relational model, and which can be realized by a  probabilistic local hidden-variable model. 

\item We characterize this construction of probabilistic models from relational ones by a \emph{maximal entropy property}, expressed in terms of a factorization of the joint distribution on measurements and outcomes into a \emph{measurement prior}, together with a family of probability distributions on outcomes conditioned on measurements. The measurement prior is also significant in allowing a fully general description of how quantum systems give rise to relational models.

\item We also give an example to show how results can be lifted from the relational setting to apply to probabilistic models. In particular, our relational versions of the GHZ and Hardy theorems lift directly to show that there are quantum systems which cannot be realized by any probabilistic local hidden-variable model.

\item We  give precise definitions of a number of classes of models. This sets the stage for developing a structure theory of these classes, which looks promising as a means for gaining insight into quantum mechanics, and both sub- and super-quantum theories. We prove a strict hierarchy result: local hidden-variable models are properly included in models arising from quantum systems, which are properly included in models satisfying No-Signalling. The latter result makes use of a relational analogue of a Popescu-Rohrlich box \cite{popescu1994quantum,khalfin1992quantum}.

\item We also consider the computational aspects of these classes of finite structures. We show that 
membership of the class of local hidden-variable models is in $\NP$. 
We also show that membership of the class of models arising from quantum systems of a given dimension is in $\PSPACE$, by reduction to the existential theory of real-closed fields.
These results suggest a number of interesting questions concerning the exact complexity of these classes.
\end{itemize}

\section{Preliminaries}

We begin by formulating the relational setting we will work in.

The systems we will consider will each have an \emph{arity} $n$, a positive integer. The arity expresses the number of \emph{parts}, which may be thought of variously as agents, sites, or subsystems, of the system under consideration. Thus a system of arity 2 is usually referred to as bipartite, and the parts are conventionally labelled as Alice and Bob.

For each part $i$, two sets are specified: a set $M_i$ of kinds of measurement which can be performed at $i$; and a set $O_i$ of possible outcomes of these measurements.
Thus for a system of arity $n$, we can form the sets
\[ M = \prod_{i=1}^n M_i, \qquad O = \prod_{i=1}^n O_i . \]
An element $\mb = (m_1 , \ldots , m_n) \in M$ specifies a choice of measurement for each part; and similarly $\ob = (o_1 , \ldots , o_n) \in O$ specifies an outcome at each part.

A \emph{system type} is given by a pair $(M, O)$ of this form.
A \emph{relational empirical model} of type $(M, O)$ is specified by a relation $e \subseteq M \times O$. As usual, such a relation can either be viewed as a set of tuples $(\mb, \ob)$, or as a characteristic function
\[ e : M \times O \rarr \{ 0, 1 \} . \]
We write $e(\mb, \ob)$ to indicate that $(\mb, \ob) \in e$, or equivalently that $e(\mb, \ob) = 1$.

Our reading of $e(\mb, \ob)$ is that, if the measurements in $\mb$ are performed, the outcomes in $\ob$ are \emph{possible} (may be observed). Of course, we do not preclude that given measurements may have more than one possible outcome.

Thus this representation of systems behaviour might be called `possibilistic' (as opposed to a \emph{probabilistic} representation which would assign probabilities to the various outcomes, conditioned on the measurements).\footnote{See the interesting essay \cite{fritz2009possibilistic} on Possibilistic Physics, which was kindly brought to my attention by Tobias Fritz after a previous version of the present paper was  made available on the arxiv.} However, it is important to note that we are using the standard logic of relations, while `Possibilistic Logic' has an established usage in Artificial Intelligence \cite{dubois3possibilistic} which is quite different.

These systems are called \emph{empirical} because they specify relationships between quantities\footnote{E.g.  settings of knobs or switches for the measurements and pointer readings for the outcomes.} all of which are directly observable.

Now we turn to hidden variables. A hidden-variable model of type $(M, O)$ has an additional set $\Lambda$ which gives the possible values of some `hidden' (unobservable) variable. The model is specified by a relation 
\[ h \subseteq M \times O \times \Lambda . \]
We say that a hidden-variable model $h$ \emph{realizes} an empirical model $e$ of the same type  if
\[ \forall \mb, \ob . \, [ e(\mb, \ob) \IFF \exists \lambda \in \Lambda. \, h(\mb, \ob, \lambda) ] . \]
Two hidden-variable models are equivalent if they realize the same empirical model in this fashion.

\subsection*{Notation}
Some notations will be helpful in allowing properties of models to be expressed succinctly.
We write $\nn = \{ 1, \ldots , n \}$. For each $i \in \nn$, we define
\[ \Mim = M_1 \times \cdots \times M_{i-1} \times M_{i+1} \times \cdots \times M_n \]
and similarly for $\Oim$. Given $m \in M_i$ and 
\[ \mb = (m_1, \ldots , m_{i-1}, m_{i+1}, \ldots , m_{n})  \in \Mim , \]
we write $m, \mb$ for the tuple
\[ (m_1, \ldots , m_{i-1}, m, m_{i+1}, \ldots , m_{n}) \]
and similarly for tuples of outcomes.

In the remainder of the paper, we shall use the notation $\mb_i$ to mean $m_i$ where $\mb = (m_1, \ldots , m_n)$; and similarly for $\ob_i$.

Given $\mb \in M$, we define
\[ e(\mb)\da \; \equiv \; \exists \ob. \, e(\mb, \ob) . \]
More generally, if $\bar{s}$ is any subsequence of arguments, we define
\[ e(\bar{s})\da  \; \equiv \; \exists \bar{t}. \, e(\bar{s}, \bar{t})  \]
where $\bar{s}, \bar{t}$ is an expansion of $\bar{s}$ to a full list of arguments for $e$. Similar notation will be used for hidden-variable models.

Note that, if $\bar{s}$ is a subsequence of $\bar{s}'$, then $e(\bar{s}')\da$ implies $e(\bar{s})\da$.

\paragraph{Example}
As an example of the notation, consider the expression 
\[ e(m, \mb, o)\da , \]
where $m \in M_i$, $\mb \in M_i^-$, $o \in O_i$. This expression expands to the following:
\[ \exists \ob \in O_i^-. \, e(m,\mb,o,\ob) . \]

\section{Properties of Models}

We now formulate a number of properties of models. These properties are, for the most part, relational versions of properties which have been discussed in the extensive literature analyzing the No-Go theorems of Quantum Mechanics, especially the Bell and Kochen-Specker theorems. Our own treatment is based to a large extent on the careful discussion,  in a unified formalism of probabilistic models, in \cite{brandenburger2008classification}. However, we shall give more emphasis to properties of empirical models.

\subsection{Properties of Empirical Models}

We shall assume a given system type $(M, O)$ of arity $n$, and formulate properties as conditions on empirical models $e$ of this type.

\subsubsection{Weak Determinism (WD)}
This says that for given measurements $\mb$, the outcomes $\ob$ are uniquely determined:
\[ \forall \mb, \ob, \ob'. [ e(\mb, \ob) \AND e(\mb, \ob') \IMP \ob = \ob'  ] . \]
In more familiar terms, it says that the relation $e$ is a \emph{partial function} from measurements to outcomes.

\subsubsection{Strong Determinism (SD)}
Strong Determinism requires that, for each $i \in \nn$,  the outcome at $i$ is uniquely determined  by the measurement at $i$: 
\[ \forall i \in \nn, \mb, \mb', \ob, \ob'. [ e(\mb, \ob) \AND e(\mb', \ob') \AND \mb_i = \mb'_i \IMP \ob_i = \ob'_i ] . \]

\subsubsection{No-Signalling (NS)}
The no-signalling condition is that the choice of measurement by one party cannot be signalled to the other parties. If we interpret the arity of a system type as implying some distributed structure, so that the different parts may be space-like separated,  it can be seen as an important residue of causality, which is needed to ensure consistency with special relativity.\footnote{However, as we shall see later (cf.~the discussion in Section~\ref{KSsec}), this is not the only possible reading, and there are some surprising connections to \emph{Contextuality}.}

This condition is usually defined in a probabilistic context by saying that the marginal probability of an outcome at $i$ for a given measurement at $i$ is independent of the other measurements. We can define No-Signalling for relational models as follows.

For all $i \in \nn$, $m \in M_i$, $o \in O_i$, $\mb, \mb' \in \Mim$:
\[  e(m, \mb, o)\da \AND e(m, \mb')\da \IMP e(m, \mb', o)\da  . \]
This says precisely that whether the outcome $o$ is possible at $i$ for a given measurement $m$ at $i$ is independent of the other measurements.

Note that there is a stronger version of this principle, which says that joint outcomes for any subset of the parts is independent of the measurements made in the remaining parts. We shall not use this version in the present paper. We refer to the treatment in \cite{ab11} which gives a more general account including the stronger form of no-signalling.

\subsection{Properties of Hidden-Variable Models}
We begin by reformulating the definitions of two of the properties we have specified for empirical models to apply to hidden variable models $h \subseteq M \times O \times \Lambda$:
\begin{description}
\item[Weak Determinism (WD)] 
\[ \forall \mb, \ob, \ob', \lambda. [ h(\mb, \ob, \lambda) \AND h(\mb, \ob', \lambda) \IMP \ob = \ob'  ] . \]
\item[Strong Determinism (SD)]
\[ \forall i \in \nn, \mb, \mb', \ob, \ob', \lambda. [ h(\mb, \ob, \lambda) \AND h(\mb', \ob', \lambda) \AND \mb_i = \mb'_i \IMP \ob_i = \ob'_i ] . \]

\end{description}

\noindent We shall now discuss some properties which have been considered specifically for hidden-variable models.

\subsubsection{Single-Valuedness (SV)}
This simply says that $\Lambda$ is a singleton. This is a rather artificial property, but it is occasionally useful.

\subsubsection{$\lambda$-Independence ($\lambda$I)}
For all $\mb, \mb' \in M$, $\lambda \in \Lambda$:
\[  h(\mb')\da \AND h(\mb, \lambda)\da \IMP h(\mb', \lambda)\da  . \]
We can read this condition as saying that the value of the hidden variable is independent of the choice of measurements.

\subsubsection{Outcome-Independence (OI)}
For all $\mb \in M$, $i \in \nn$, $o, o'  \in O_i$, $\ob, \ob' \in \Oim$, $\lambda \in \Lambda$:
\[ h(\mb, o, \ob, \lambda) \AND 
h(\mb, o', \ob', \lambda) \IMP h(\mb, o, \ob', \lambda)  . \]
This says that, for given measurements and value of the hidden variable, the possibility of an outcome at $i$ is independent of which other outcomes occur.

This is easily seen to be logically equivalent to the following condition:
\begin{equation}
\label{OIeq}
\forall \mb, \ob, \lambda. \, [ h(\mb, \ob, \lambda) \IFF \bigwedge_{i=1}^n h(\mb, \ob_i, \lambda)\da ]. 
\end{equation}

\subsubsection{Parameter-Independence (PI)}
This is essentially  the reformulation of No-Signalling for hidden-variable models.

For all $i \in \nn$, $m \in M_i$, $o \in O_i$, $\mb, \mb' \in \Mim$, $\lambda \in  \Lambda$: 
\[  h(m, \mb, o, \lambda)\da \AND h(m, \mb', \lambda)\da \IMP h(m, \mb', o, \lambda)\da  . \]
It says that, conditional on the value of the hidden variable, the possible outcomes of  a measurement at $i$ are independent of the other measurements.

\subsubsection{Locality (L)}
The assumption of locality  is that the possible outcomes of a measurement, for a given value of the hidden variable, are locally determined, in the sense that the outcome at $i$ depends only on the measurement performed at $i$. This is expressed as follows:
\[ \forall \mb, \ob, \lambda . \, [ h(\mb, \lambda)\da \AND \bigwedge_{i=1}^n h(\mb_i, \ob_i, \lambda)\da \IMP h(\mb, \ob, \lambda) ] . \]

\section{Implications}

We shall now consider which implications hold between these conditions.

\begin{theorem}
\label{hvimpsprop}
The following implications hold between properties of hidden-variable models:
\begin{enumerate}
\item Weak Determinism implies Outcome Independence.
\item Strong Determinism is equivalent to the conjunction of Weak Determinism and Parameter Independence.
\item Locality is equivalent to the conjunction of Parameter Independence and Outcome Independence.
\end{enumerate}
\end{theorem}
\begin{proof}
1. We assume Weak Determinism and prove Outcome Independence in the equivalent form (\ref{OIeq}).
Assume that for all $i \in \nn$, $h(\mb, o_i, \lambda)\da$. This means that for all $i$, there is $\ob^{(i)} \in \Oim$ such that $h(\mb, o_i, \ob^{(i)}, \lambda)$. By \WD, we conclude that $o_1, \ob^{(1)} = o_2, \ob^{(2)} = \cdots = o_n,  \ob^{(n)}$, and so $h(\mb, \ob, \lambda)$, as required.

2. Assume Strong Determinism. That this implies Weak Determinism  is immediate from the definitions. To prove \PI, suppose that $h(m, \mb, o, \lambda)\da$ and $h(m, \mb', \lambda)\da$, so for some $o'$, $h(m, \mb', o', \lambda)\da$. By \SD, $o = o'$, so 
\[ h(m, \mb', o, \lambda)\da , \]
as required.

Now assume \WD\ and \PI. Suppose that $h(m, \mb, \ob, \lambda)$ and $h(m, \mb', \ob', \lambda)$, where $m \in M_i$ and $\mb, \mb' \in \Mim$. 
Let $o = \ob_i$. Then $h(m, \mb, o, \lambda)\da$ and $h(m, \mb', \lambda)\da$, so by \PI, $h(m, \mb', o, \lambda)\da$, so that $h(m, \mb', o, \ob'', \lambda)$ for some $\ob'' \in \Oim$. By \WD, $\ob_i = o = \ob'_i$.

3. Assume \Loc.
To prove \OI, suppose that $h(\mb, \ob_i, \lambda)\da$ for all $i$. This implies that $h(\mb_i, \ob_i, \lambda)\da$ for all $i$, and also that $h(\mb,\lambda)\da$, and hence by \Loc \ that $h(\mb, \ob, \lambda)$, as required. 

\noindent To prove \PI, suppose that $h(m, \mb, o, \lambda)\da$ and $h(m, \mb', \lambda)\da$. This implies that $h(m, o, \lambda)\da$, and  for all $j \neq i$, for some $o_j$, $h(\mb'_{j}, o_j, \lambda)\da$. Hence by $\Loc$, $h(m, \mb', o, \lambda)\da$.

Now assume \OI\ and \PI. Suppose that $h(\mb_{i}, \ob_{i}, \lambda)\da$, $i \in \nn$, and $h(\mb, \lambda)\da$. This implies that for each $i$, for some $\mb^{(i)} \in M_i^-$, $h(\mb_{i}, \mb^{(i)}, \ob_{i}, \lambda)\da$. Applying \PI\ $n$ times, we obtain that $h(\mb_{i}, \mb_{-i}, \ob_{i}, \lambda)\da$ for each $i$, where $\mb_{i}, \mb_{-i} = \mb$.
Hence for all $i \in \nn$, $h(\mb, \ob_{i}, \lambda)\da$, and applying \OI, we obtain $h(\mb, \ob, \lambda)$, which proves that \Loc\ holds.
\end{proof}

\begin{corollary}
\label{SDLoccorr}
Strong Determinism implies Locality.
\end{corollary}

We now show a relationship between properties of a hidden-variable model and the induced empirical model. Let $h \subseteq M \times O \times \Lambda$ be a hidden-variable model;
 we define the \emph{induced} empirical model $e$ to be the (unique) model realized by $h$:
\[ e(\mb, \ob) \; \equiv \; \exists \lambda \in \Lambda . \, h(\mb, \ob, \lambda) . \]

\begin{proposition}
\label{LIPIprop}
If $h$ satisfies $\lambda$-Independence and Parameter Independence, then $e$ satisfies No-Signalling.
\end{proposition}
\begin{proof}
Suppose that $e(m, \mb, o)\da$ and $e(m, \mb')\da$. Then for some $\lambda$, 
\[ h(m,\mb,o, \lambda)\da , \]
and so $h(m, \mb, \lambda)\da$, and also $h(m, \mb')\da$. By \LI, $h(m, \mb',\lambda)\da$. Hence by \PI, $h(m, \mb', o, \lambda)\da$, and so $e(m, \mb', o)\da$.
\end{proof}

We can strengthen this as follows. 

\begin{proposition}
\label{NSequivprop}
An empirical model satisfies No-Signalling if and only if it can be realized by a hidden-variable model satisfying $\lambda$-Independence and Parameter Independence.
\end{proposition}
\begin{proof}
One direction is Proposition~\ref{LIPIprop}. For the converse, if $e$ satisfies \NS, then the unique hidden-variable model satisfying \SV\ which induces it trivially satisfies \LI\ and \PI.
\end{proof}

\section{Existence}
In this section, we show some positive results of the form that every empirical 
model can be realized by a hidden-variable model satisfying certain properties. Of course, not every combination of properties is possible; this will be the content of the No-Go theorems to follow.

\begin{proposition}
\label{SVprop}
Every empirical model is realized by a hidden-variable model satisfying Single-Valuedness.
\end{proposition}
\begin{proof}
Immediate.
\end{proof}

\begin{proposition}
Every empirical model is realized by a hidden-variable model satisfying Strong Determinism.
\end{proposition}
\begin{proof}
Given $e \subseteq M \times O$, we define
\[ \Lambda = \{ \Phi \subseteq e \mid \Phi = \phi_1 \times \cdots \times \phi_n, \phi_i : M_i \pfn O_i, i \in \nn \} . \]
The hidden-variable model is defined as follows:
\[ h(\mb, \ob, \Phi) \; \equiv \; (\Phi(\mb) = \ob) . \]
We must show that the empirical model induced by $h$ is $e$.
Suppose that $e(\mb, \ob)$. Then we can define $\phi_i = \{ (\mb_i, \ob_i) \}$, $i \in \nn$, and $\Phi = \phi_1 \times \cdots \times \phi_n$. Then $h(\mb, \ob, \Phi)$, so $(\mb, \ob)$ is in the induced relation.
For the converse, suppose that $h(\mb, \ob, \Phi)$. Then $\Phi(\mb) = \ob$, and since $\Phi \subseteq e$, $e(\mb, \ob)$.

Now we show that $h$ satisfies \SD. Suppose that $h(\mb, \ob, \Phi)$, $h(\mb', \ob', \Phi)$, and $\mb_i = \mb'_i$. Then $\ob_i = \phi_i(\mb_i) = \phi_i(\mb'_i) = \ob'_i$.
\end{proof}

\begin{proposition}
\label{WDprop}
Every empirical model is realized by a hidden-variable model satisfying Weak Determinism and $\lambda$-Independence.
\end{proposition}
\begin{proof}
Given $e \subseteq M \times O$, we define $\dom(e) = \{ \mb \mid e(\mb)\da \}$, and
\[ \Lambda = \{ \Phi \subseteq e \mid \Phi : \dom(e) \rarr O \} . \]
Thus the values of the hidden variable are \emph{choice functions}, which select outcomes for each choice of measurements which has a non-empty set of outcomes.
The hidden-variable model is defined by:
\[ h(\mb, \ob, \Phi) \; \equiv \; (\Phi(\mb) = \ob) . \]
We must show that the empirical model induced by $h$ is $e$. Suppose that $e(\mb, \ob)$. Then there is some choice function $\Phi$ such that $\Phi(\mb) = \ob$, and $h(\mb, \ob, \Phi)$. The converse is proved as for the previous Proposition.

To show that $h$ is Weakly Deterministic, suppose that $h(\mb, \ob, \Phi)$ and $h(\mb, \ob', \Phi)$. Then $\ob = \Phi(\mb) = \ob'$.

Finally, suppose that  $h(\mb, \Phi)$ and  $h(\mb')\da$. Then $\mb' \in \dom(e)$, and so $h(\mb', \Phi)\da$. This shows that $h$ satisfies \LI.
\end{proof}

We now show that if an empirical model can be realized by a hidden-variable model satisfying \LI\ and \Loc, then it can be realized by one satisfying \LI\ and \SD. Note that in general, \SD\ is strictly stronger than \Loc\ if these properties are considered in isolation.

This is a relational analogue of a result proved for probabilistic models in \cite{fine1982hidden,BK}.

\begin{proposition}
\label{LILocSDprop}
Let $h$ be a hidden-variable model satisfying \LI\ and \Loc. There is an equivalent model $h'$ satisfying \LI\ and \SD.
\end{proposition}
\begin{proof}
We are given $h \subseteq M \times O \times \Lambda$ satisfying \LI\ and \Loc.
We define
\[ \Lp = \{ \lambda \in \Lambda \mid h(\lambda)\da \}, \quad \Mp = \{ \mb \in M \mid h(\mb)\da  \} , \]
and for $\mb \in M$, $\lambda \in \Lambda$:
\[ \Oml = \{ \ob \in O \mid h(\mb, \ob, \lambda) \} . \]
Since $h$ satisfies \LI, for any $\mb \in \Mp$ and $\lambda \in \Lp$, $\Oml \neq \vn$.

We also define local versions of these notions, for each $i \in \nn$:
\[ \Mip = \{ m \in M_i \mid h(m)\da \}, \]
and for $\lambda \in \Lp$ and $m \in \Mip$,
\[ \Oiml = \{ o \in O_i  \mid h(m, o, \lambda)\da \} .\]
Since $h$ satisfies \LI, for any $m \in \Mip$ and $\lambda \in \Lp$, $\Oiml \neq \vn$.
Moreover, since $h$ satisfies \Loc, we have
\begin{equation}
\label{OLoceq}
\Oml = \prod_{i \in \nn}  O^{i}_{\mb_i, \lambda} .
\end{equation}
Indeed,
\[ \ob \in \Oml  \equiv h(\mb, \ob, \lambda) \siff \bigwedge_{i \in \nn} h(\mb_i, \ob_i, \lambda)\da  \equiv \bigwedge_{i \in \nn} \ob_i \in O^{i}_{\mb_i, \lambda} \siff \ob \in \prod_{i \in \nn}  O^{i}_{\mb_i, \lambda} . \]
Note that $\Mp \subseteq \prod_{i=1}^{n} \Mip$, but in general we need \emph{not} have equality.
We will return to this point in section~\ref{probmodsec}.

\noindent Now we define a new value space for hidden variables.
This is most elegantly expressed as a \emph{dependent type} \cite{barendregt1992lambda}:
\[ \Lv \; = \; (\sum \lambda \in \Lp)(\prod i \in \nn)(\prod m \in \Mip) \, \Oiml . \]
Explicitly, $\Lv$ consists of pairs $(\lambda, \Phi)$, where $\lambda \in \Lp$ and $\Phi = (\Phi_1, \ldots , \Phi_n)$, where $\Phi_i : \Mip \rarr O_i$, such that $\Phi_i(m) \in \Oiml$.

We define a new hidden variable model $h' \subseteq M \times O \times \Lv$ by:
\[ h'(\mb, \ob, (\lambda, \Phi)) \; \equiv \;  (\mb \in \Mp \AND \bigwedge_{i \in \nn} \Phi_i(\mb_i) = \ob_i ). \]
If $h'(\mb, \ob, (\lambda, \Phi))$, then $\mb \in \Mp$ and $\lambda \in \Lp$.
By construction, $\ob_i \in O^{i}_{\mb_i, \lambda}$, so by (\ref{OLoceq}), $\ob \in \Oml$, and $h(\mb, \ob, \lambda)$.
Conversely, if $h(\mb, \ob, \lambda)$, then $\mb \in \Mp$ and $\lambda \in \Lp$. For each $i \in \nn$, $\mb_i \in \Mip$, and $\ob_i \in O^{i}_{\mb_i, \lambda}$. 
Since $\Oiml$ is non-empty for each $m \in \Mip$ and $\lambda \in \Lp$, we can define $\Phi_i \in (\prod m \in \Mip)\, \Oiml$ with $\Phi_i(\mb_i) = \ob_i$. Hence there is $\Phi =  (\Phi_1, \ldots , \Phi_n)$ such that $h'(\mb, \ob, (\lambda, \Phi))$. Thus $h'$ is equivalent to $h$.

If  $h'(\mb, \ob, (\lambda, \Phi))$ and $h'(\mb', \ob',  (\lambda, \Phi))$ and $\mb_i = \mb'_i$, then 
$\ob_i = \Phi_i(\mb_i) = \Phi_i(\mb'_i) = \ob'_i$, so 
$h' $ satisfies \SD.

Now suppose that $h'(\mb')\da$ and $h'(\mb, (\lambda, \Phi))\da$. Then $\mb' \in \Mp$ and $\lambda \in \Lp$, and since $h$ satisfies \LI, $h(\mb', \lambda)\da$. Since $\mb' \in \Mp$,  $\Phi$ is defined on $\mb'$, and
so $h'(\mb', (\lambda, \Phi))\da$. Thus $h'$ satisfies \LI.

\end{proof}

\section{No-Go Results}

We shall now prove a number of results showing that there are empirical models which cannot be realized by any hidden-variable model with certain prescribed properties. These results are based directly on four classic constructions in the foundations of quantum mechanics: EPR \cite{einstein1935can}, GHZ \cite{greenberger1990bell}, Hardy \cite{hardy1993nonlocality}, and Kochen-Specker \cite{kochen1975problem}. However, our treatment is carried out entirely in our simple relational framework. Our versions of these constructions involve only finite sets and relations.

\subsection{EPR}

Our first result is by nature of a warm-up, following \cite{brandenburger2008classification}.

\begin{proposition}
\label{EPRprop}
There is an empirical model which is not realized by any hidden-variable model satisfying \SV\ and \OI.
\end{proposition}
\begin{proof}
We let $M_1 = \{ X \}$, $M_2 = \{ Y \}$, $O_1 = O_2 = \{ a, b \}$. We define $e$ by
\begin{center}
\begin{tabular}{c|cccc} 
$e$ & $(a, a)$ & $(a, b)$ & $(b, a)$ & $(b, b)$ \\ \hline
$(X, Y)$ &  $0$ & $1$ & $1$ & $0$ \\
\end{tabular}
\end{center}

\noindent Assume that $h$ satisfies \SV\ and \OI, with $\Lambda = \{ \lambda \}$. Suppose for a contradiction that $h$ realizes $e$. Then $h(X, Y, a, b, \lambda)$ and $h(X, Y, b, a, \lambda)$, so by \OI, $h(X, Y, a, a, \lambda)$. This implies $e(X, Y, a, a)$, yielding the required contradiction.
\end{proof}
Note that we can always find a hidden-variable model  realizing $e$ which satisfies \SV, by Proposition~\ref{SVprop}, and also one satisfying \OI, by Propositions~\ref{WDprop} and \ref{hvimpsprop}(1).
In fact, because $M$ is a singleton, the Weakly Deterministic model constructed for $e$ by Proposition~\ref{WDprop} is actually Strongly Deterministic. Applying Proposition~\ref{hvimpsprop} again, we see that it satisfies \Loc. Trivially, it satisfies \LI. Thus the `relational EPR model' does have a local hidden-variable model.

\subsection{GHZ}

We define a system type of arity 3 by:
\[ M_i = \{ 1, 2 \}, \quad O_i = \{ R, G \}, \qquad i = 1, 2, 3. \]
To lighten the notation, we shall write $122$ rather than $(1, 2, 2)$, and similarly for other tuples.
Let $P = \{ 122, 212, 221 \} \subseteq M$.
We consider any empirical relation $e$ such that:
\[ \forall p \in P. \; e(p) = \{  RRR, RGG, GRG, GGR \}  \]
\[ e(111) = \{ RRG, RGR, GRR, GGG \} . \]
Here we treat the relation $e$ in its equivalent form as a set-valued function $e : M \rarr  \Pow(O)$, so that $e(\mb) = \{ \ob \mid e(\mb, \ob) \}$..

Thus $e$ is completely specified on $P \cup \{ 111 \}$. It can have arbitrary behaviour on other measurements. We call any such $e$ a \emph{GHZ model}.

\begin{proposition}
\label{GHZprop}
No GHZ model $e$ can be realized by a hidden-variable model satisfying \LI\ and \Loc.
\end{proposition}
\begin{proof}
Assume for a contradiction that there is a hidden-variable model which satisfies  \LI\ and \Loc, and realizes $e$. Applying Proposition~\ref{LILocSDprop}, this implies that there is a hidden-variable model $h$ which satisfies \LI\ and \SD\ and realizes $e$.

The assumption that $h$ realizes $e$
implies that for some $\lambda$, $h(111, \lambda)\da$, and also $h(p)\da$ for $p \in P$. By \LI, this implies that $h(p, \lambda)\da$ for each $p \in P$.

We now analyze the set of outcomes $\ob$ such that $h(122, \ob, \lambda)$. Since $h$ realizes $e$, this must be a non-empty subset $S$ of $\{ RRR, RGG, GRG, GGR \}$. Since $h$ satisfies \SD, it must be a singleton. Similar reasoning applies to $T  = \{ \ob \mid h(221, \ob, \lambda)\}$.

Since $122$ and $221$ have the same middle measurement, and $h$ satisfies \SD, if $h(122, \ob, \lambda)$ and $h(221, \ob', \lambda)$, we must have $\ob_{2} = \ob'_{2}$. Hence if $S = \{ RRR \}$ or $S = \{ GRG \}$, then $T = \{ RRR \}$ or $T = \{ GRG \}$, and  if $S = \{ RGG \}$ or $S = \{ GGR \}$, then $T = \{ RGG \}$ or $T = \{ GGR \}$.
Thus we have 8 possible joint assignments under $\lambda$ to $122$ and $221$.
We can now check that any of these completely determines the assignment to $212$.

Suppose for example that $122 \mapsto RRR$, $221 \mapsto GRG$. Then by \SD, $212 \mapsto G{-}R$, and the only consistent possibility for the middle outcome is that $212 \mapsto GGR$. We can represent this joint assignment to the measurements in $P$ as the `instruction'  
\[ \inst{RGG}{GRR}  \]
The rubric is that the $i$'th row gives the outcomes under $\lambda$ when the measurements are set to $i$, $i = 1, 2$.
A similar analysis applies to the other seven cases.

We can tabulate this well-known `Mermin instruction set'  \cite{mermin1990quantum} as follows:
\[ \inst{RRR}{RRR} \quad \inst{RGG}{RGG} \quad \inst{GRG}{GRG} \quad \inst{GGR}{GGR} \]
\[ \inst{RGG}{GRR} \quad \inst{RRR}{GGG} \quad \inst{GGR}{RRG} \quad \inst{GRG}{RGR} \]

\noindent For any of these cases, let $\ob$ be the top row of the instruction, and let $\mb = 111$. We have $h(\mb_i, \ob_i, \lambda)\da$, $i = 1, \ldots , 3$, and hence by \Loc, we have $h(111, \ob, \lambda)$. Since each top row contains an odd number of $R$'s, and the possible outcomes for $111$ under $e$ all contain an even number of $R$'s, we obtain the desired contradiction to the assumption that $h$ realizes $e$.
\end{proof}

\subsubsection*{Discussion}
We have followed Mermin's classic presentation of the GHZ argument \cite{mermin1990quantum} closely. However, compelling and polished as his account is, it does not fully specify the precise assumptions which are being used. In particular, his use of `instructions' tacitly assumes \emph{determinism}, which is made plausible on physical grounds. 
We replace this tacit assumption by Proposition~\ref{LILocSDprop}, which shows that if a local hidden-variable realization exists, there must be one satisfying 
Strong Determinism.

The mathematical content of the  `instructions' which appear in the proof will be explained in generality in Proposition~\ref{unionsdprop}.

\subsection{Hardy}

We shall now give a relational formulation of the `Hardy paradox' \cite{hardy1993nonlocality}. The result is similar to the one based on the GHZ construction, but is of additional interest for several reasons:
\begin{itemize}
\item It applies to bipartite systems, whereas GHZ is essentially (at least) tripartite.
\item The Hardy construction avoids inequalities, but is probabilistic in character. However, the argument can be carried out in purely relational or possibilistic terms. 
\item The construction leads to a family of  `axioms' which must be satisfied by all models which can be realized by local hidden variables. This has some flavour of a logical version of the CHSH inequalities \cite{clauser1969proposed}, as suggested in \cite{fritz2009possibilistic}.
\end{itemize}

Our formulation follows the lines of \cite{mermin1994quantum,fritz2010quantum}, although in our opinion  the present treatment is clearer and more explicit as to exactly which assumptions are being used.

We shall be concerned with bipartite systems of the following type:
\[ M_1 = \{ X_1, X_2 \}, \quad M_2 = \{ Y_1, Y_2 \}, \qquad O_1 =   O_2 = \{ R, G \} . \]
We consider relational models $e$ satisfying the following condition:
\begin{center}
\begin{tabular}{c|cccc} 
 & $(R, R)$ & $(R,G)$ & $(G,R)$ & $(G, G)$ \\ \hline
$(X_1, Y_1)$ &  $1$ &  &  &  \\
$(X_1, Y_2)$ &   $0$ &  &  &  \\
$(X_2, Y_1)$ &  $0$ &  &  & \\
$(X_2, Y_2)$ &   &  &  & $0$ \\
\end{tabular}
\end{center}
What this means is that $e$ must take the specified values; no condition is being imposed on the remaining entries. We shall also assume that $e$ is \emph{total}, meaning that every measurement combination has some possible outcome. We shall call models satisfying these two conditions \emph{Hardy models}.

\begin{proposition}
\label{hardyprop}
No Hardy model can be realized by a hidden-variable model satisfying \LI\ and \Loc.
\end{proposition}
\begin{proof}
Assume for a contradiction that there is a hidden-variable model which satisfies  \LI\ and \Loc, and realizes $e$. Applying Proposition~\ref{LILocSDprop}, this implies that there is a hidden-variable model $h$ which satisfies \LI\ and \SD\ and realizes $e$.
By Proposition~\ref{hvimpsprop}(3) and Corollary~\ref{SDLoccorr}, $h$ also satisfies \PI.

Since $h$ realizes $e$, for some $\lambda$, $h(X_1 Y_1, RR, \lambda)$. Since $e$ is total, $h(X_1 Y_2)\da$, and by \LI, $h(X_1 Y_2, \lambda)\da$. By \PI, for some $o$ we must have 
$h(X_1 Y_2, Ro, \lambda)$.  Since $e(X_1 Y_2, RR)$ is excluded by one of the Hardy conditions, and $h$ realizes $e$ by assumption, we must have $h(X_1 Y_2, RG, \lambda)$. Similar reasoning now shows that we must have $h(X_2 Y_2, RG, \lambda)$, and $h(X_2 Y_1, RG, \lambda)$. However, $h$ is strongly deterministic, and from  $h(X_2 Y_1, RG, \lambda)$ and $h(X_1 Y_1, RR, \lambda)$ this implies $R = G$, yielding the required contradiction.
\end{proof}

\subsubsection*{Discussion}
Mermin's discussion of this result in \cite{mermin1994quantum} makes some play of the fact that, unlike his presentation of  the GHZ argument,  Einsteinian `elements of reality' and Merminian `instructions' do not appear. However, we can see that in fact exactly the same assumptions are required for a rigorous proof. The elegance of the Hardy result is that it applies to the bipartite case, and to a class of models satisfying simple general conditions. The \emph{quantum realization} of these models does involve probabilities strictly between 0 and 1, as we shall see in the next Section; however, the relational formulation allows us to see clearly that the No-Go result itself is purely logical in character.\footnote{Indeed, the quantum realization of the GHZ construction also necessarily involves probabilities strictly between 0 and 1, as must any probability distribution whose possibilistic collapse is a many-valued relation. In either case, an experimental verification that a physical system does have the specified behaviour of  a GHZ or Hardy model --- and hence has no realization by hidden variables --- will require many runs of the system.}

An immediate corollary of Proposition~\ref{hardyprop} is that a necessary condition for (total) relational models to have local hidden-variable realizations is that they satisfy
the implication
\[ a \IMP b \OR c \OR d \]
where $a$, $b$, $c$, $d$ are boolean variables placed at the indicated points of the table:
\begin{center}
\begin{tabular}{c|cccc} 
 & $(R, R)$ & $(R,G)$ & $(G,R)$ & $(G, G)$ \\ \hline
$(X_1, Y_1)$ &  $a$ &  &  &  \\
$(X_1, Y_2)$ &   $b$ &  &  &  \\
$(X_2, Y_1)$ &  $c$ &  &  & \\
$(X_2, Y_2)$ &   &  &  & $d$ \\
\end{tabular}
\end{center}
Since $1$ corresponds to logical truth, the fact that this implication holds can be written as the boolean inequality
\[ 1 \leq \neg a \OR b \OR c \OR d . \]
Note also that there are 8 variants of this constraint, arising from the symmetries
\[ X_1 \leftrightarrow X_2, \qquad Y_1 \leftrightarrow Y_2 , \qquad R \leftrightarrow G . \]
Whether a useful  theory of logical conditions characterizing local hidden-variable and other classes of models can be developed in the relational setting, paralleling the use of inequalities on probabilistic models to define correlation polytopes \cite{pitowsky1989quantum}, remains to be seen. Some work in this direction is reported in \cite{fritz2009possibilistic}.
\subsection{KS}
\label{KSsec}

We now turn to a result based on the Kochen-Specker theorem \cite{kochen1975problem}. We follow \cite{brandenburger2008classification}, and give a proof based on the 18-vector construction in 4 dimensions of \cite{cabello1996bell}, although of course our account is purely in terms of discrete sets and relations. 

We begin with some preliminary notions.
Given a system of arity $n$ and type $(M, O)$, the symmetry group $S_n$ acts on $M$ and $O$ in the evident fashion:
\[ \begin{array}{rcl}
\pi \cdot (m_1, \ldots , m_n)  & = & (m_{\pi^{-1}(1)}, \ldots , m_{\pi^{-1}(n)}), \\
\pi \cdot (o_1, \ldots , o_n) & = & (o_{\pi^{-1}(1)}, \ldots , o_{\pi^{-1}(n)}) .
\end{array}  \]
We say that an empirical relation $e \subseteq M \times O$ is \emph{equivariant} (or `satisfies Exchangability', to use probabilistic terminology as in \cite{brandenburger2008classification}) if, for all $\pi \in S_n$:
\[ e(\mb, \ob) \IFF e(\pi \cdot \mb, \pi \cdot \ob) . \]

\begin{proposition}
\label{KSprop}
There is an empirical model which cannot be realized by any hidden-variable model satisfying \LI\ and \PI.
\end{proposition}
\begin{proof}
We shall use a system of arity 4, with the following type:
\[ M_i = \{ m_1, \ldots , m_{18} \}, \quad O_i = \{ 0, 1 \}, \quad i = 1, \ldots , 4. \]
Consider the table (from \cite{brandenburger2008classification,cabello1996bell}):
\begin{center}
\begin{tabular}{|c|c|c|c|c|c|c|c|c|} \hline
$m_1$ & $m_1$ & $m_8$ & $m_8$ & $m_2$ & $m_9$ & $m_{16}$ & $m_{16}$  & $m_{17}$ \\ \hline
$m_2$ & $m_5$ & $m_9$ & $m_{11}$ & $m_5$ & $m_{11}$ & $m_{17}$ & $m_{18}$ & $m_{18}$  \\ \hline
$m_3$ & $m_6$ & $m_3$ & $m_7$ & $m_{13}$ &  $m_{14}$ & $m_4$ & $m_6$ & $m_{13}$  \\ \hline
$m_4$ & $m_7$ & $m_{10}$ & $m_{12}$ & $m_{14}$  & $m_{15}$ & $m_{10}$ & $m_{12}$ & $m_{15}$  \\ \hline
\end{tabular}
\end{center}
Let $P \subseteq M$ be the set of quadruples of measurements corresponding to the columns of this table. Let $Q \subseteq O$ be the set 
\[ \{ (1, 0, 0, 0), (0, 1, 0, 0), (0, 0, 1, 0), (0, 0, 0, 1) \} . \]
The specification of the empirical model $e \subseteq M \times O$ is as follows:
\begin{enumerate}
\item For some function $f : P \rarr Q$,  for all $\mb \in P$:
\[ e(\mb, \ob) \IFF f(\mb) = \ob . \]
\item $e$ is equivariant.
\end{enumerate}
Note that this specification can always be met, since we can take an arbitrary function $f$ as in (1), and then expand its definition to fulfill equivariance. In particular, no element of $P$ is a permutation of any other element, so there is no conflict between (1) and (2).

Now assume for a contradiction that $e$ is realized by a hidden-variable model $h$ satisfying \LI\ and \PI. By Proposition~\ref{LIPIprop}, this implies that $e$ satisfies No-Signalling. Consider the assignment $e$ makes to the first column. Let the (unique) element in this column assigned 1 by $e$ be $m_i$. Note that every element appears in exactly 2 columns in the table. By equivariance, the assignment $e$ makes to the other column $j$ in which $m_i$ appears is unchanged if we permute the elements in that column so that $m_i$ appears in the same row as it does in column 1. But then by \NS, it must be the case that $m_i$ is also the unique element assigned 1 in column $j$. The same argument can be applied to every column, and we conclude that the function $f$ must be \emph{non-contextual}; that is, it assigns the same value, 0 or 1,  to each $m \in \{ m_1, \ldots , m_{18} \}$ regardless of where it appears as a component in $P$. However, no such $f$ can exist, since each of the nine columns is assigned exactly one 1, so an odd number of 1's appears among the outcomes assigned to the elements of $P$; while each $m$ appears twice in $P$, so the number of 1's arising from any non-contextual assignment must be even. This yields the required contradiction.
\end{proof}

\subsubsection*{Discussion}
While formally the above argument is clear-cut (and of course follows \cite{cabello1996bell} and \cite{brandenburger2008classification}), conceptually there are some surprises. In particular:
\begin{itemize}
\item The fact that No-Signalling gives rise to non-contextuality is unexpected.
\item The Kochen-Specker theorem is of course meant to show the contextuality of quantum mechanics. Thus the empirical system we used in the proof of Proposition~\ref{KSprop} should arise from quantum mechanics --- and indeed it does, as we shall recall in the next section. However, quantum  systems are supposed to satisfy No-Signalling! So what is going on?
\end{itemize}
The answer to the second point is that in this case, arities are \emph{not} being interpreted as standing in for spatially distributed structure, so issues of causality do not arise. Rather, as we shall see in more detail in the next section, the arities in this case really correspond to different \emph{observables}, which may be applied to the \emph{same} (e.g.~single-particle) system. The `measurements' in the system type then correspond to different `branches' of these measurements --- \ie the projectors arising in their spectral decompositions. In this context, `No-Signalling' means the ability to swap outcomes for given projectors between different measurements --- \ie non-contextuality.

We might say that under this reading, we are considering constraints on information flow between  \emph{different measurements of the same system},  as counterfactual alternatives, rather than between space-like separated regions.

It is not clear if the fact that Kochen-Specker and GHZ can be brought into a common format in this way can lead to a more unified understanding of these phenomena.
It is tempting to look for a connection with Mermin's proposal for a unified derivation of KS and GHZ \cite{mermin1990simple}. However, on closer inspection this does not appear to be related. Note in particular that KS as it appears here (following \cite{brandenburger2008classification}) is a substantially \emph{different} result to GHZ, since KS excludes a much wider range of hidden-variable models.

\section{Physical Models}
\label{physmodsec}

So far, our formal development has been physics-free. We see this clean separation between a simple mathematical framework, which can be used to prove clear-cut results from precisely formulated assumptions, and the much more complex structures and concepts  from physics which provide the motivation, as a virtue. However, it is of course important to make the connection, which we shall now do. 

\subsection*{Empirical Models Arising from Quantum Mechanics}

We shall now spell out how quantum  systems give rise to a class of empirical models, which we shall call $\QM$. Thus if $\EM$ is the class of all empirical models, $\QM \subseteq \EM$.

Suppose we are given a system type $(M, O)$ of arity $n$.
A \emph{quantum realization} of this system type is specified by the following data:
\begin{itemize}
\item Finite dimensional Hilbert spaces $\HH_1, \ldots , \HH_n$.
\item For each $i \in \nn$, $m \in M_i$, and $o \in O_i$, a linear operator $A_{m,o}$ on $\HH_i$, subject to the condition:
\[ \sum_{o \in O_i} A_{m, o}^{\dagger} A_{m, o} = I_{\HH_i} . \]
Thus $(A_{m, o})_{o \in O_i}$ forms a \emph{generalized measurement} \cite{nielsen2000quantum}. 
\item A state $\rho$, \ie a density operator on  $\HH_1 \otimes \cdots \otimes H_n$. 
\end{itemize}
For each choice of measurement $\mb \in M$, and outcome $\ob \in O$, the usual `statistical algorithm' of quantum mechanics  defines a probability $p_{\mb}(\ob)$ for obtaining outcome $\ob$ from performing the measurement $\mb$  on $\rho$:
\[ p_{\mb}(\ob) = \mathsf{Tr}(\Amo^{\dagger} \Amo  \rho) \]
where $\Amo = A_{\mb_1, \ob_1} \otimes \cdots \otimes A_{\mb_n, \ob_n}$.

We  define a relational empirical model $e \subseteq M \times O$ by
\[ e(\mb, \ob) \; \equiv \; (p_{\mb}(\ob) > 0) . \]
Thus $e$ arises as the `possibilistic collapse' of the usual quantum mechanical formalism.
We take $\QM$ to be the class of empirical models which are realized by  quantum systems in this fashion.

In the examples to follow, we will be in the special case where $\rho$ is a pure state, $\rho = \ket{\psi}\bra{\psi}$, and the measurements are projective, so $A_{m, o} = \ket{\psi_{m,o}}\bra{\psi_{m, o}}$, where
$\psi_{m, o}$ is the eigenvector corresponding to the outcome $o$. In this case, the statistical algorithm is equivalently  given by:
\[ p_{\mb}(\ob) = | \langle \psi \mid \psi_{\mb, \ob} \rangle |^2 , \]
where $ \psi_{\mb, \ob}  =  \psi_{\mb_1, \ob_1} \otimes \cdots \otimes \psi_{\mb_n, \ob_n}$.

We can now add to the content of the results of the previous section by indicating how the empirical models used in the proofs arise from quantum systems; that is, showing that  these systems are in the class $\QM$. This will `complete' the usual derivations of these results in application to quantum mechanics, although it is striking how much of the work can be factored out to the purely relational level.

\subsubsection*{EPR}

The empirical system used in the proof of Proposition~\ref{EPRprop} arises from the 2-qubit system, $\HH_1 = \HH_2 = \Comp^2$, with the state $\frac{| 01  \rangle + | 10 \rangle}{\sqrt{2}}$, and  with 1-qubit measurements in the computational basis.

\subsubsection*{GHZ}
\label{GHZqmsubsec}

Consider a 3-qubit system, $\HH_1 = \HH_2 = \HH_3 = \Comp^2$. The `GHZ state'  is represented as 
\[ \frac{| 000  \rangle + | 111 \rangle}{\sqrt{2}} \]
in the computational basis.
We can interpret the computational basis  in each component as corresponding to the measurement for spin Up or Down along the $z$-axis. We also have measurement for spin Right or Left along the $x$-axis, with basis vectors 
\[ \frac{| 0  \rangle + | 1 \rangle}{\sqrt{2}}, \qquad \frac{| 0  \rangle - | 1 \rangle}{\sqrt{2}} \]
and similarly for spin Forward or Back along the $y$-axis, with basis vectors
\[ \frac{| 0  \rangle + i | 1 \rangle}{\sqrt{2}}, \qquad \frac{| 0  \rangle - i | 1 \rangle}{\sqrt{2}} \]
These bases, with eigenvalues corresponding to the spins, determine observables $X$ and $Y$.
For example, $\frac{| 0  \rangle + | 1 \rangle}{\sqrt{2}}$ is the eigenvector of $X$ corresponding to the measurement outcome spin Right.
For each component, we interpret the measurements $1$ and $2$ as $X$ and $Y$, and the outcomes $G$ and $R$ as the two possible spin directions along the given axis: $G$ for spin Right for $X$ and spin Forward for $Y$,  and $R$ for the alternative outcomes. We can compute the quantum mechanical probabilities for these measurement outcomes on the GHZ state; for example, the probability distribution on outcomes for the measurement $XYY$, corresponding to $122$ in our labelling, is:
\[ p_{XYY}(\ob) = \left\{ \begin{array}{ll}
1/4, & \ob \in \{  RRR, RGG, GRG, GGR \} \\
0 & \mbox{otherwise.}
\end{array} \right.
\]
The possibilistic collapse of these probabilities produces a GHZ model,
as used in the proof of Proposition~\ref{GHZprop}.

\subsubsection*{Hardy}

A detailed discussion of quantum realizations of the Hardy construction is given in \cite{hardy1993nonlocality,mermin1994quantum}.
We shall just give a simple concrete instance.

We consider the two-qubit system, with $X_2$ and $Y_2$ measurement in the computational basis.
We take $R=0$, $G=1$.
The eigenvectors for $X_1$ are taken to be
\[ \sqrt{\frac{3}{5}} \ket{0}  + \sqrt{\frac{2}{5}} \ket{1} , \qquad  - \sqrt{\frac{2}{5}} \ket{0}  + \sqrt{\frac{3}{5}} \ket{1}  \]
and similarly for $Y_1$. The state is taken to be
\[ \sqrt{\frac{3}{8}} \ket{10} \; + \; \sqrt{\frac{3}{8}} \ket{01} \; - \; \frac{1}{2} \ket{00} . \]
One can then calculate the probabilities to be
\[ p_{X_1 Y_2}(RR) \; = \; p_{X_2 Y_1}(RR) \; = \; p_{X_2 Y_2}(GG ) = 0, \]
and $p_{X_1 Y_1}(RR) = 0.09$, which is very near the maximum attainable value \cite{mermin1994quantum}. The possibilistic collapse of this model is thus a Hardy model.

\subsubsection*{KS}

For the Kochen-Specker argument, as previously explained in our discussion in \ref{KSsec}, the empirical model does not arise from a quantum system in the fashion we have been describing.
Rather, the elements $m_1, \ldots , m_{18}$ are taken to be unit vectors in $\Real^4$, chosen so that each quadruple in the table used in the proof forms an orthonormal basis. A suitable assignment  of vectors is given in \cite{cabello1996bell}.

Thus we see that there is indeed no conflict with the usual understanding of the behaviour of physical systems, and in particular that they satisfy No-Signalling.

\section{Probabilistic Models}
\label{probmodsec}

Given a finite system type $(M, O)$, let $p : M \times O \rarr [0, 1]$ be a probability distribution.
We can form a relational model $e$ as the \emph{possibilistic collapse} of $p$:
\[ e(\mb, \ob) \; \equiv \; p(\mb, \ob) > 0 . \]
Probabilistic models such as $p$ are the empirical models studied in \cite{brandenburger2008classification}. Similarly, the hidden-variable models studied there are probability distributions
\[ q : M \times O \times \Lambda \rarr [0, 1] . \]
We can form the possibilistic collapse of these to relational hidden-variable models $h \subseteq M \times O \times \Lambda$ in similar fashion.

The properties of models we have been studying in this paper are all relational analogues of properties of probabilistic models. So it is natural to ask: which properties of probabilistic models are inherited by their possibilistic collapses?

Firstly, some notation. Given a probability distribution $p$ and an incomplete list of arguments $s$, $p(s)$ is obtained as the marginal, \ie by summing over all extensions of $s$ to a full list of arguments. Similar conventions apply to conditional probabilities.
With this convention, we note that $p(s) > 0$ iff $e(s)\da$, where $e$ is the possibilistic collapse of $p$.

The standard definition of No-Signalling for probabilistic systems is as follows.
\begin{itemize}
\item Probabilistic No-Signalling (\PNS).

For all $m$, $o$, $\mb$, $\mb'$:
\[ p(m, \mb) > 0 \AND p(m, \mb') > 0 \IMP p(o | m, \mb) = p(o | m, \mb') . \]
\end{itemize}

\begin{proposition}
\label{posspreserveNSprop}
If $p$ satisfies \PNS, then its possibilistic collapse satisfies \NS.
\end{proposition}
\begin{proof}
Let $e$ be the possibilistic collapse of $p$. Suppose that $e(m, \mb, o)\da$ and $e(m, \mb')\da$.
Then $p(m, \mb) > 0$ and $p(m, \mb') > 0$, and by \PNS, $p(o | m, \mb) = p(o | m, \mb')$. Now 
\[ e(m, \mb, o)\da \IMP p(o | m, \mb) > 0 \IMP p(o | m, \mb') > 0 \IMP e(m, \mb', o)\da .  \qed \]
\end{proof}

\noindent We now consider properties of probabilistic hidden-variable models. These are the notions which have been studied in the literature, and which we have found relational analogues for in this paper.
We shall formulate these under the assumption that the set of values $\Lambda$ for the hidden variable is finite. This incurs no loss of generality, given that the measurement and outcome sets are finite, by virtue of Lemma~6.5 of \cite{BK}.

\begin{itemize}
\item Probabilistic $\lambda$-Independence (\PLI) \cite{dickson1999quantum}.

For all $\mb$, $\mb'$, $\lambda$:
\[ q(\mb) > 0 \AND q(\mb') > 0 \IMP q(\lambda | \mb) = q(\lambda | \mb') . \]

\item Probabilistic Outcome Independence (\POI) \cite{jarrett1984physical,shimony1986events}.

For all $\mb$, $o \in O_i$, $\ob \in O_i^-$, $\lambda$:
\[ q(\mb, \lambda) > 0 \IMP q(o |  \mb, \lambda) = q(o | \ob, \mb, \lambda) . \]

\item Probabilistic Parameter Independence (\PPI) \cite{jarrett1984physical,shimony1986events}.

For all $o$, $m \in M_i$, $\mb \in M_i^-$, $\lambda$:
\[ q(\mb, \lambda) > 0 \IMP q(o |m,  \mb, \lambda) = q(o | m, \lambda) . \]

\item Probabilistic Locality (\PLoc) \cite{bell1964einstein}.

For all $\mb$, $\ob$, $\lambda$:
\[ q(\mb, \lambda) > 0 \IMP \, q(\ob | \mb, \lambda) = \prod_{i=1}^n q(\ob_i | \mb_i , \lambda) . \]
\end{itemize}

\begin{proposition}
\label{Probpresprop}
For each property $\varphi \in \{ \mbox{\LI}, \mbox{\OI}, \mbox{\PI}, \mbox{\Loc} \}$:
if $q$ satisfies $\mathbf{P}\varphi$, then its possibilistic collapse satisfies $\varphi$.
\end{proposition}

\noindent This is proved similarly to Proposition~\ref{posspreserveNSprop}.
\subsection*{An Example}

We shall now show by an illustrative example how the results of the last few sections can be put together to lift No-Go theorems from the relational setting to apply to probabilistic models.

The following definition is taken from \cite{brandenburger2008classification}\footnote{The terminology `$q$ is equivalent to $p$' is used in \cite{brandenburger2008classification}.}. A probabilistic hidden variable model $q : M \times O \times \Lambda \rarr [0, 1]$ realizes a probabilistic empirical model $p : M \times O \rarr [0, 1]$  if for all $\mb$, $p(\mb) > 0 \IFF q(\mb) > 0$, and
\[ p(\mb) > 0 \IMP \forall \ob. \, p(\ob | \mb) = q(\ob | \mb) . \]

\begin{proposition}
\label{probrealprop}
If $q$ realizes $p$, then the possibilistic collapse of $q$ realizes the possibilistic collapse of $p$.
\end{proposition}

Next, we observe that the process of obtaining a relational model from a quantum system, as described in section~\ref{physmodsec}, naturally factors into two parts; obtaining a probabilistic model, and then applying the possibilistic collapse.

Now we can lift Propositions~\ref{GHZprop} and~\ref{hardyprop} to results about probabilistic models.
\begin{proposition}
There is a probabilistic empirical model $p$ which arises from a quantum system, and which is not realized by any probabilistic hidden-variable model satisfying \PLI\ and \PLoc.
\end{proposition}
\begin{proof}
We let $p$ be the probabilistic model arising from the GHZ system described in section~\ref{physmodsec}. Suppose  that $p$ is realized by a probabilistic hidden variable model $q$.
Now assume for a contradiction that $q$ satisfies \PLI\ and \PLoc. By Proposition~\ref{Probpresprop}, the possibilistic collapse of $q$ is a relational hidden-variable model $h$ satisfying \LI\  and \Loc, and by Proposition~\ref{probrealprop}, $h$ realizes the empirical relational model $e$ arising as the possibilistic collapse of $p$. Since $e$ meets the conditions of Proposition~\ref{GHZprop}, this yields the required contradiction.

Exactly the same argument can be made for Hardy models, using Proposition~\ref{hardyprop}.
\end{proof}

\section{Building Probabilistic Models from Relational Models}

We now turn to the question of some kind of converse result. Of course, the possibilistic collapse loses a great deal of information. For example, the probabilistic systems of arity 1 with one measurement and $n$ outcomes form the probability simplex $\Delta(n)$ on $n$ vertices. Their image under possibilistic collapse is the set of non-empty subsets of $\nn$. In general, if the possibilistic collapse of $p$ satisfies a property $\varphi$, it need not be the case that $p$ satisfies $\mathbf{P}\varphi$.

We can, however, ask a different question, which is important in relating the respective structure of classes of  probabilistic and relational models:
if we have a relational model $h$ satisfying some property  $\varphi$, can we find a probabilistic model $q$ whose possibilistic collapse is $h$, and which satisfies $\mathbf{P}\varphi$? In other terms, we already know that possibilistic collapse preserves properties of models, so its image on the class of probabilistic models satisfying $\mathbf{P}\varphi$ is included in the class of relational models satisfying $\varphi$. We are asking if it is \emph{surjective}.

We shall obtain strikingly different results depending on the property in question.
\begin{itemize}
\item For No-Signalling, we shall prove a negative result: there are relational models satisfying No-Signalling which do not arise from any probabilistic model satisfying Probabilistic No-Signalling.
This holds even for bipartite systems.

\item For local hidden-variable models, we shall obtain a positive result. Moreover, the constructions involved in showing this are conceptually interesting in their own right, and give a new twist to the old question of assigning probabilities to logical structures.
\end{itemize}

\subsection{No-Signalling: Negative Result}

\begin{proposition}
\label{nsignegprop}
There are empirical models $e$ satisfying \NS\ such that there is no probabilistic model $p$ satisfying \PNS\ whose possibilistic collapse is $e$.
\end{proposition}
\begin{proof}
We shall give an explicit counter-example. Consider the following bipartite system, with
\[ M_1 = \{ X_1, X_2 \}, \quad M_2 = \{ Y_1, Y_2 \}, \quad O_1 = \{ a_1, a_2 \}, \quad O_2 = \{ b_1, b_2 \} . \]
We tabulate $e$ as a $4 \times 4$ boolean matrix:
\begin{center}
\begin{tabular}{c|cccc} 
 & $(a_1, b_1)$ & $(a_1, b_2)$ & $(a_2, b_1)$ & $(a_2, b_2)$ \\ \hline
$(X_1, Y_1)$ &  $1$ & $1$ & $0$ & $1$ \\
$(X_1, Y_2)$ &  $1$ & $0$ & $1$ & $1$ \\
$(X_2, Y_1)$ &  $1$ & $0$ & $1$ & $1$ \\
$(X_2, Y_2)$ &  $1$ & $1$ & $0$ & $1$ \\
\end{tabular}
\end{center}
One can check by calculation that $e$ satisfies \NS. 
Now assume for a contradiction that there is a probabilistic model satisfying \PNS\ whose possibilistic collapse is $e$. Such a model has the form
\begin{center}
\begin{tabular}{c|cccc} 
 & $(a_1, b_1)$ & $(a_1, b_2)$ & $(a_2, b_1)$ & $(a_2, b_2)$ \\ \hline
$(X_1, Y_1)$ &  $c$ & $d$ & $0$ & $e$ \\
$(X_1, Y_2)$ &  $f$ & $0$ & $g$ & $h$ \\
$(X_2, Y_1)$ &  $i$ & $0$ & $j$ & $k$ \\
$(X_2, Y_2)$ &  $l$ & $m$ & $0$ & $n$ \\
\end{tabular}
\end{center}
where $c, \ldots , n$ are positive real numbers. Applying the equations arising from \PNS\, we obtain:
\[ \begin{array}{lcl}
c+d & = & f \\
e & = & g + h \\
i & = & l + m \\
j+k & = & n \\
c & = & i + j \\
d + e & = & k \\
f + g & = & l \\
h & = & m + n
\end{array}
\]
Since these are positive numbers, this implies
\[  c < f < l < i < c \]
and hence $c < c$, yielding the required contradiction.

As another example, we give a somewhat more compact
tripartite system, with
$M_1 = \{ X \}$, $M_2 = \{Y \}$, $M_3 = \{Z_1, Z_2 \}$, $O_1 = \{ a_1, a_2 \}$, $O_2 = \{ b_1 , b_2 \}$, $O_3 = \{ c \}$.
We tabulate the relation  as follows:
\begin{center}
\begin{tabular}{c|cccc}
&  $(a_1, b_1, c)$ & $(a_1, b_2, c)$ & $(a_2, b_1, c)$ & $(a_2, b_2, c)$   \\ \hline
$(X, Y, Z_1)$ & $1$ & $1$ & 0 & $1$  \\
$(X, Y, Z_2)$ & $1$ & 0 & $1$ & $1$  \\
\end{tabular}
\end{center}
Similar reasoning can be applied to this example.
\end{proof}

\subsection{Locality: Positive Results}

\subsection*{Measurement Locality}

We will need to consider one additional property of models, which as far as we know has not been discussed previously in the literature. In the probabilistic format, following \cite{brandenburger2008classification}, we are using \emph{joint distributions}, on measurements as well as outcomes (and hidden parameters in the case of hidden-variable models). This means that probabilities $p(\mb)$ are being assigned to measurements themselves; whereas in most accounts, one only considers a family of conditional probabilities $p(\ob | \mb)$, for outcomes conditioned on measurements. Of course, the joint distribution does determine these conditional probabilities, and can be considered more general. However, some issues of interpretation arise; see \cite{brandenburger2008classification} for a discussion.

In particular, when we consider locality conditions, it seems natural to suppose that which measurements may be selected at site $i$ and lead to an outcome should be independent of the measurement choices made at other sites. This is expressed in probabilistic form by the following property:
\begin{itemize}
\item Probabilistic Measurement Locality (\PML):
\[ \forall \mb \in M. \, p(\mb) = \prod_{i=1}^{n} p(\mb_i) . \]
\end{itemize}

The relational form is:
\begin{itemize}
\item Measurement Locality (\ML):
\[ \forall \mb \in M. \, e(\mb)\da \IFF \bigwedge_{i=1}^{n} e(\mb_i)\da . \]
\end{itemize}

For a hidden-variable model, the condition is:
\[ \forall \mb \in M. \, h(\mb, \lambda)\da \IFF \bigwedge_{i=1}^{n} h(\mb_i, \lambda)\da . \]
It is easy to see that possibilistic collapse of a system satisfying \PML\ results in a system satisfying \ML.

Note that, given a family of  probability distributions $p_{\mb}$ on $O$, we can always pass to a joint distribution on $M \times O$, by setting
\[ p(\mb, \ob) = \frac{p_{\mb}(\ob)}{N} \]
where $N = \card{M}$. Of course, we then recover the probabilities $p_{\mb}(\ob) = p(\ob|\mb)$ we started with from this assignment; while $p(\mb) = 1/N$. Moreover, this joint distribution will satisfy $\PML$.

\subsection*{The Construction}
Given  a relational hidden-variable model $h \subseteq M \times O \times \lambda$ satisfying \LI, the probabilistic hidden-variable model $q^h$ is defined by:
\[ q^{h}(\mb, \ob, \lambda) \; =\; \left\{ \begin{array}{ll}
1/ W_{\mb, \lambda} & \quad h(\mb)\da \\
0 & \quad \mbox{otherwise,}
\end{array}  \right.
\]
where 
\[ L = \card{\{ \lambda \in \Lambda \mid h(\lambda)\da \}}, \qquad N = \card{\{ \mb \mid h(\mb)\da \}}, \]\[  \Klm = \card{\{ \ob \mid h(\mb, \ob, \lambda) \} } , \quad W_{\mb, \lambda} = \Klm L N . \]
It is clear that the possibilistic collapse of $q^h$ is $h$.
Also, $q^h$ is a probability distribution: since $h$ satisfies \LI, $h(\lambda)\da$ and $h(\mb)\da$ implies $h(\mb, \lambda)\da$, and hence
\[ \sum_{\lambda, \mb, \ob} q^h(\mb, \ob, \lambda) \;\;  = \;\; LN(\Klm / \Wlm) \;\; = \;\; 1. \]

\begin{proposition}
\label{relprobprop}
Let $h \subseteq M \times O \times \Lambda$ be a relational hidden-variable model satisfying  \ML, \LI, and \Loc. Then $q^h$ satisfies \PML, \PLI , and \PLoc.
\end{proposition}
\begin{proof}
For each $i \in \nn$, we define:
\[ \begin{array}{lcl}
N^i & = & \card{\{ m \in M_i  \mid h(m)\da \}} \\
K^{i}_{m, \lambda} & = & \card{\{ o \in O_i \mid h(m, o, \lambda) \da \} } .
\end{array}
\]
Since by assumption, $h$ satisfies \ML, 
\[ h(\mb, \lambda)\da \IFF \bigwedge_{i \in \nn} h(\mb_i, \lambda)\da ,\]
and since $h$ satisfies \LI,
\[ h(\mb)\da \IFF  \bigwedge_{i \in \nn} h(\mb_i)\da . \]
Similarly, since $h$ satisfies \Loc, we have
\begin{equation}
\label{Klmeq}
\forall \mb, \ob, \lambda. \, [ h(\mb, \ob, \lambda) \IFF \bigwedge_{i \in \nn} h(\mb, \ob_i, \lambda)\da ]. 
\end{equation}

\noindent Hence we have
\begin{equation}
\label{locdefseqn}
N = \prod_{i \in \nn} N^i , \qquad \Klm = \prod_{i \in \nn} \Klim  . 
\end{equation}

\noindent So the `logical' independence conditions \LI\ and \Loc\  imply that these quantities can be computed locally from the numbers $N^i$, $\Klim$.

\noindent Firstly, we consider \PLI. Suppose that $q^h(\mb) > 0$ and $q^h(\mb') > 0$. Then $h(\mb)\da$ and $h(\mb')\da$. We evaluate $q^h(\lambda | \mb)$:
\[ q^h(\lambda | \mb) \;\; = \;\; \frac{\sum_{\ob} q^h(\mb, \ob, \lambda)}{\sum_{\ob, \lambda'} q^h(\mb, \ob, \lambda')} . \]
The numerator evaluates to $\Klm / \Wlm = 1/LN$; while the denominator evaluates to
\[ \sum_{\lambda' \in \Lambda} \frac{K_{\mb, \lambda'}}{W_{\mb, \lambda'}} \;\; = \;\; \frac{L}{LN} \;\; = \;\; \frac{1}{N} . \]
Hence $q^h(\lambda | m) = 1/L$, which is independent of $\mb$.

\noindent Now we prove \PLoc. Suppose that $q^h(\mb, \lambda) > 0$, which implies $h(\mb, \lambda)\da$. Firstly, we calculate:
\[ q^h(\ob | \mb, \lambda)  \;\; = \;\;  \frac{q^h(\mb, \ob, \lambda)}{\sum_{\ob'} q^h(\mb, \ob', \lambda)}  \;\; = \;\;  \frac{\Wlm}{\Klm \Wlm}   \;\; = \;\;  \frac{1}{\Klm} . \]
So by (\ref{locdefseqn}), it suffices to show, for each $i \in \nn$, $m \in M_i$ with $h(m)\da$, $o \in O_i$:
\[ q^h(o|m, \lambda) \;\; = \;\;  \frac{1}{K^{i}_{m, \lambda}} . \]
We have:
\[ q^h(o|m, \lambda) \;\; = \;\; \frac{\sum_{\ob, \mb} q^h(m, \mb, o, \ob, \lambda)}{\sum_{o', \ob, \mb} q^h(m, \mb, o', \ob, \lambda)} . \]
Evaluating the numerator, we obtain:
\[ \sum_{\mb} \frac{K^{-i}_{m, \mb, \lambda}}{W_{m, \mb, \lambda}} \;\; = \;\;  \frac{1}{K^{i}_{m, \lambda} L N^{i}} , \]
where 
\[ K^{-i}_{m, \mb, \lambda}  \;\; = \;\;  \prod_{j \neq i} K^{j}_{m,\mb_{j}, \lambda} . \]

\noindent For the denominator, we have
\[ \sum_{\mb} \sum_{o', \ob} \frac{1}{W_{m, \mb, o', \ob, \lambda}}  \;\; = \;\;  \sum_{\mb} \frac{K_{m,\mb, \lambda}}{K_{m,\mb, \lambda}LN}  \;\; = \;\;  \sum_{\mb}  \frac{1}{LN} \;\; = \;\; \frac{N^{-i}}{LN} \;\; = \;\; \frac{1}{LN^{i}} . \]
Dividing through, we obtain $1/K^{i}_{m, \lambda}$, and the proof is complete.
\end{proof}


\section{Maximum Entropy Characterization}
The construction of probabilistic models from relational ones 
which we have described in the previous section 
is clearly canonical in some sense --- but which? 
We now turn to the task of answering this question.

Clearly, it is some form of maximum entropy construction, where probability is shared as evenly as possible, subject to respecting the structure of the locality conditions. It is, however, too naive to expect the joint probability distribution $p : M \times O \rarr [0, 1]$ itself to have maximum entropy among all such with the same possibilistic collapse. The distribution with maximum entropy among all those with a given possibilistic collapse $e$ will simply be the uniform distribution on the support 
\[ \{ (\mb, \ob) \in M \times O \mid e(\mb, \ob) \} . \]
This makes no distinction between measurements and outcomes, and does not respect the local structure. 

Thus we must find a more subtle characterization. To do this, we firstly need to sharpen our discussion of probabilistic models.
The following elementary observation will be useful.
\begin{proposition}
\label{jointdecompprop}
Let $X$, $Y$ be finite sets. Given a probability distribution $p : X \times Y \rarr [0, 1]$, we can form a 
probability distribution $\tp$ on $X$ by marginalization:
\[ \tp(x) = \sum_{y \in Y} p(x, y) . \]
We can also form a probability distribution $p_x : y \mapsto p(x,y)/\tp(x)$ on $Y$  for each $x \in X$ such that $\tp(x) > 0$. In more familiar terms, $p_x$ is the marginalization of the conditional probability $p(\cdot|x)$. Conversely, given a probability distribution $\tp$ on $X$, and a family $\{ p_{x} \}_{x \in X'}$ of probability distributions on $Y$, where $X' = \{x \in X \mid \tp(x) > 0 \}$,  we can form a joint distribution $p'$ on $X \times Y$ by:
\begin{equation}
\label{margeq}
p'(x, y) = \left\{ \begin{array}{ll}
\tp(x) p_{x}(y), & \tp(x) > 0, \\
0, & \mbox{\textrm{otherwise}.} 
\end{array}
\right. 
\end{equation}
These passages are mutually inverse.
\end{proposition}
\begin{proof}
If we start with $p$, pass to $\tp$ and $\{ p_x \}$, and then form a joint distribution $p'$ from these by~(\ref{margeq}), then firstly
\[ \tp(x) > 0 \IFF \sum_y p(x, y) > 0 . \]
Thus if $\tp(x) = 0$, then $p(x, y) = 0$. If $\tp(x) > 0$, then
\[ p'(x, y) = \tp(x)p_{x}(y)  = \frac{\tp(x) p(x, y)}{\tp(x)} = p(x, y) . \]
Thus $p' = p$.

Conversely, if we start with $\tp$ and $\{ p_x \}$, pass to $p'$, and then to $\theta_{p'}$ and $\{ p'_x \}$, then
\[ \theta_{p'}(x) = \sum_y p'(x, y) = \sum_y  \tp(x) p_x(y) = \tp(x) . \]
If $\tp(x) > 0$, then
\[ p'_{x}(y)  = \frac{p'(x,y)}{\tp(x)} = \frac{\tp(x)p_x(y)}{\tp(x)} = p_x(y) . \]
\end{proof}

\noindent In our context, we take $X = M$ and $Y = O$ for probabilistic empirical models, and $X = M \times \Lambda$, $Y = O$ for hidden-variable models.
We refer to the distribution $\tp$ as the \emph{measurement prior}. It can be understood as an initial probability on a given combination of measurements being performed. Whether this arises from some experimental protocol involving randomization, from ignorance, or from some other source, is not specified.

We recall the standard definition of the \emph{entropy} of a probability distribution $p : Z \rarr [0, 1]$ on a finite set $Z$:
\[ H(p) = - \sum_z p(z) \log p(z) . \]
The following lemma gives a useful  formula for the entropy of a joint distribution following the decomposition of Proposition~\ref{jointdecompprop}.
\begin{lemma}
Let  $p : X \times Y \rarr [0, 1]$ be a probability distribution, $X$, $Y$ finite. Then
\[ H(p) = H(\tp) \;+\; \sum_{\tp(x) > 0} \tp(x) H(p_x) . \]
\end{lemma}
\begin{proof}
By calculation:
\begin{eqnarray*}
H(p) & = & {-} \sum_{\tp(x)>0, y} p(x,y) \log p(x,y) \\
& = &  {-} \sum_{\tp(x)>0, y} \tp(x) p_{x}(y) \log \tp(x) p_{x}(y) \\
& = & {-} \sum_{\tp(x)>0, y} \tp(x) p_{x}(y) (\log \tp(x)  + \log p_{x}(y)) \\
& = & {-} \sum_{\tp(x)>0, y} \tp(x) p_{x}(y) \log \tp(x) \; + \;  {-} \sum_{\tp(x)>0, y} \tp(x) p_{x}(y)\log p_{x}(y) \\
& = & {-} \sum_{\tp(x)>0} (\sum_y p_{x}(y))  \tp(x) \log \tp(x)  \; + \;  {-} \sum_{\tp(x)>0, y} \tp(x) p_{x}(y)\log p_{x}(y) \\
& = &  {-} \sum_{\tp(x)>0} \tp(x) \log \tp(x) \; + \;   \sum_{\tp(x)>0} \tp(x) ({-} \sum_{y} p_{x}(y)\log p_{x}(y)) \\
& = &  H(\tp) \; + \; \sum_{\tp(x) > 0} \tp(x) H(p_x) . 
\end{eqnarray*}
\end{proof}

\begin{proposition}[Maximum Entropy Characterization]
Given a hidden-variable model $h$ satisfying \LI, let $C$ be the set of probabilistic models whose possibilistic collapse is $h$.  For all $q \in C$:
\[ H(\thq) \geq H(\tq) \;\; \mbox{and} \;\; [h(\mb, \lambda)\da \IMP H(q^{h}_{\mb, \lambda}) \geq H(q_{\mb, \lambda})  ]. \]
\end{proposition}
\begin{proof}
Given $q \in C$, we note firstly that since $q$ and $q^h$ both have the same possibilistic collapse, $\thq$ and $\tq$ have the same support, $\{ (\mb, \lambda) \in M \times \Lambda \mid e(\mb, \lambda)\da \}$. Also, if $(\mb, \lambda)$ is in the support, $\thq(\mb, \lambda) = 1/LN$. Thus $\thq$ is the uniform distribution on the support, and hence $H(\thq) \geq H(\tq)$.

Now given $(\mb, \lambda)$ in the support of $\thq$, $q^{h}_{\mb,\lambda}(\ob) = 1/ \Klm$, so $q^{h}_{\mb,\lambda}$ is the uniform distribution on its support, $\{ \ob \in O \mid h(\mb, \ob, \lambda) \}$. Again, since $q$ has the same possibilistic collapse, $q_{\mb, \lambda}$ has the same support as $q^{h}_{\mb,\lambda}$. Hence $H(q^{h}_{\mb, \lambda}) \geq H(q_{\mb,\lambda})$ for all $(\mb,\lambda)$ in the support of $\thq$.
\end{proof}

\section{Classes of Models}

The class of relational empirical models  can be seen as a simplified image of the space of physical theories. As we have seen, a surprising amount of the structure of locality and related notions is preserved under the passage by possibilistic collapse to this image. Thus it seems promising that by understanding the structure of this class, we can gain insight into quantum mechanics, and both sub- and super-quantum theories.

All the examples of empirical models we have considered so far have been of finite type; it seems these suffice to encapsulate the key structural and conceptual issues which arise in the foundations of quantum mechanics. We define $\EM$ to be the class of empirical models of finite type. We define $\LHV$ to be the class of empirical models $e$ which are realized by hidden-variable models satisfying \LI\ and \Loc.
We also have the class $\NSig$ of empirical models which satisfy No-Signalling. In the light of Proposition~\ref{nsignegprop}, we also define the class $\NSp$ of models which arise as the possibilistic collapse of probabilistic models satisfying \PNS.

The class $\QM$ comprises  those empirical models which arise from quantum systems in the manner described in  section~\ref{physmodsec}. However, this description needs to be refined in one respect. The models arising from quantum systems naturally give rise to probability distributions on outcomes $p_{\mb}$. To be completed into joint distributions on $M \times O$ according to the recipe provided by Proposition~\ref{jointdecompprop}, they also need to be equipped with a measurement prior.
Since we are only concerned here with the relational models they give rise to by possibilistic collapse, this simply amounts to restricting the set of possible measurements.

Given a set $S \subseteq M$ and an empirical model $e \subseteq M \times O$, we define the \emph{restriction} of $e$ to $S$:
\[ e_{S}(\mb, \ob) \; \equiv \; e(\mb, \ob) \AND \mb \in S . \]

Importantly, we have:
\begin{proposition}
If $e$ satisfies \NS, so does $e_{S}$.
\end{proposition}
\begin{proof}
A trivial consequence of the definition of \NS.
\end{proof}
We define $\QM$ to be the class of models obtained from quantum systems by possibilistic collapse and restriction to a subset of measurements.

Our goal in this section is lay some basic groundwork for future study by establishing
 the following (strict) inclusions:
\[ \LHV \subset \QM \subset \NSp \subset \NSig \subset \EM . \]

\subsection{The Quantum Non-Locality Zone}

\begin{proposition}
We have the strict inclusion $\LHV \subset \QM$.
\end{proposition}
\begin{proof}
Firstly, consider a relational model  of arity 1, with $d$ outcomes. We can represent this by a quantum system over a Hilbert space of dimension $d$, with a chosen basis $\ket{i}$. We use the maximally mixed state
\[ \mu \; = \; \sum_{i=1}^{d} \, \frac{1}{d} \ket{i}\bra{i} . \]
We can then use projectors to represent arbitrary subsets of the outcomes.

Now given a model $e$ of arity $n$ in $\LHV$, we use the product state $\mu_1 \otimes \cdots \otimes \mu_n$, where $\mu_i$ is the maximally mixed state in a chosen basis for a Hilbert space $\HH_i$ of dimension $d_i = \card{O_i}$.  By locality, this system realizes $e$. If $e$ is not total, we can use restriction to cut down the admissible measurements in the quantum system appropriately.
Thus $\LHV \subseteq \QM$.

By Proposition~\ref{GHZprop}, we know that there are systems in $\QM$ which are not in $\LHV$, so the inclusion is strict.
\end{proof}

\subsection{The Super-Quantum Sub-Luminal Zone}

We invoke a standard result: the No-Signalling or No-Communication Theorem \cite{ghirardi1980general,jordan1983quantum,kennedy1995empirical}.
\begin{theorem}
\label{NCTthm}
Every probabilistic model arising from a quantum system satisfies \PNS.
\end{theorem}

As an immediate corollary of this theorem and Propositions~\ref{NSequivprop} and~\ref{posspreserveNSprop}, we obtain:

\begin{proposition}
Every empirical model arising from a quantum system has a realization by a hidden-variable model satisfying \LI\ and \PI.
\end{proposition}

It follows from Proposition~\ref{hvimpsprop} that the missing ingredient in such hidden-variable realizations will be \OI.

A topic which has received much attention in recent work on quantum information is that of \emph{non-local boxes}; devices which exhibit super-quantum correlations. In particular, the family of Popescu-Rohrlich boxes \cite{popescu1994quantum,khalfin1992quantum} achieve maximum violation of the Tsirelson bound on quantum correlations \cite{tsirelson1980}, while still satisfying No-Signalling.

Consider the probabilistic model of arity 2, with $M_{i} = O_{i} = \{ 0, 1 \}$, $i = 1, 2$. The conditional probabilities $p(\ob | \mb)$ are tabulated as follows, where the rows are indexed by measurements, and the columns by outcomes:
\begin{center}
\begin{tabular}{c|ccccc}
  & $(0, 0)$ & $(1, 0)$ & $(0, 1)$ & $(1, 1)$  &  \\ \hline
$(0, 0)$ & $\alpha$ & 0 & 0 & $\alpha'$ & \\
$(1, 0)$ & $\beta$ & 0 & 0 & $\beta'$ & \\
$(0, 1)$ & $\gamma$ & 0 & 0 & $\gamma'$ & \\
$(1, 1)$ & 0 & $\delta$ & $\delta'$ & 0 & 
\end{tabular}
\end{center}

\begin{lemma}
\label{eqproblemm}
If the model satisfies \PNS, then all the entries labelled by Greek letters must be equal to $1/2$.
\end{lemma}
\begin{proof}
If we expand the definition of \PNS, we obtain 8 equations of the form
\[ p(o_{i} = j | m_{i} = k, m_{i'} = 0) = p(o_{i} = j | m_{i} = k, m_{i'} = 1), \]
$i =1, 2$, $i' = 3 - i$, $j,k = 0, 1$. 
Applying the equation with $i=1$, $j=k=0$ to the table, this yields $\alpha + 0 = \gamma + 0$, hence $\alpha = \gamma$. The remaining equations yield
\[ \alpha' = \gamma', \;\; \beta = \delta', \;\; \beta' = \delta,  \;\; \alpha = \beta, \;\; \alpha' = \beta', \;\; \gamma = \delta, \;\;  \gamma' = \delta'  \]
in similar fashion. Altogether, these imply that all these quantities are equal, and since each row of the table is a probability distribution, their common value must be $1/2$.
\end{proof}

Hence the only probabilistic model of this form satisfying No-Signalling is
\begin{center}
\begin{tabular}{c|ccccc}
& $(0, 0)$ & $(1, 0)$ & $(0, 1)$ & $(1, 1)$  &  \\ \hline
$(0, 0)$ & $1/2$ & 0 & 0 & $1/2$ & \\
$(1, 0)$ & $1/2$ & 0 & 0 & $1/2$ & \\
$(0, 1)$ & $1/2$ & 0 & 0 & $1/2$ & \\
$(1, 1)$ & 0 & $1/2$ & $1/2$ & 0 & 
\end{tabular}
\end{center}
Note that this model has the property that the non-zero entries are exactly those for which
\[ o_{1} \oplus o_{2} = m_{1}m_{2} \]
where $\oplus$ is addition modulo 2. This is the simplest of the 8 PR-boxes \cite{barrett2005nonlocal}; it satisfies
\[ \EE(0, 0) + \EE(1, 0) + \EE(0,1) - \EE(1,1) = 4 \]
where  
\[ \EE(x,y) = \sum_{a, b = 0}^{1} (-1)^{a+b} p(o_1 = a, o_2 = b \mid m_1 = x, m_2 = y) \] 
measures the correlation of the outcomes.
Thus this model exceeds the Tsirelson bound \cite{tsirelson1980} of $2 \sqrt{2}$ for the maximum degree of correlation that can be achieved by any bipartite quantum  system of this form. It follows that \emph{no quantum system can give rise to this probabilistic model}.

\begin{proposition}
We have the strict inclusion $\QM \subset \NSp$.
\end{proposition}
\begin{proof}
By Theorem~\ref{NCTthm}, any quantum system gives rise to a probabilistic model satisfying \PNS. By 
definition, the possibilistic collapse of this system is in $\NSp$. Thus we have the inclusion $\QM \subseteq \NSp$.

Now consider the relational model
\begin{center}
\begin{tabular}{c|ccccc}
& $(0, 0)$ & $(1, 0)$ & $(0, 1)$ & $(1, 1)$  &  \\ \hline
$(0, 0)$ & $1$ & 0 & 0 & $1$ & \\
$(1, 0)$ & $1$ & 0 & 0 & $1$ & \\
$(0, 1)$ & $1$ & 0 & 0 & $1$ & \\
$(1, 1)$ & 0 & $1$ & $1$ & 0 & 
\end{tabular}
\end{center}
By Lemma~\ref{eqproblemm}, the only probabilistic model satisfying \PNS\ which gives rise to it by possibilistic collapse is the PR box. This shows both that the model is in $\NSp$, and that \emph{no quantum system can give rise to this relational model}. 
\end{proof}

\begin{proposition}
We have the strict inclusion $\NSp \subset \NSig$.
\end{proposition}
\begin{proof}
If $e$ is in $\NSp$, it arises by possibilistic collapse from $p$ satisfying \PNS. By Proposition~\ref{posspreserveNSprop}, it follows that $e$ satisfies \NS. Thus $\NSp \subseteq \NSig$. The strictness of the inclusion follows immediately from Proposition~\ref{nsignegprop}.
\end{proof}

\subsection{And Beyond}

\begin{proposition}
We have the strict inclusion $\NSig \subset \EM$.
\end{proposition}
\begin{proof}
By definition, $\NSig \subseteq \EM$. The strictness of the inclusion follows from Proposition~\ref{KSprop}.
\end{proof}

\section{Computational Aspects}

\subsection{Hidden-Variable Models}

We shall  investigate the computational complexity of the class $\LHV_t$, of \emph{total} relational models which can be realized by local hidden-variable models. 
We begin with a simple `intrinsic' characterization of this class, which does not refer to hidden variables.

\begin{proposition}
\label{unionsdprop}
A relational model is in $\LHV_t$ if and only if it is a union of total strongly deterministic models.
\end{proposition}
\begin{proof}
Note firstly that a total model $e \subseteq M \times O$ is strongly deterministic if and only if $e = \prod_i f_i$, where $f_i : M_i \rarr O_i$, $i \in \nn$. If $e = \bigcup_{\lambda} f_{\lambda}$, where $f_{\lambda} = \prod_i f_{\lambda, i}$, we can define a hidden variable model $h$ by
\begin{equation}
\label{hvsdeq}
h(\mb, \ob, \lambda) \equiv f_{\lambda}(\mb) = \ob . 
\end{equation}
This trivially satisfies \LI, since each $f_{\lambda}$ is total, and satisfies \SD\ by construction.

Conversely, if $e$ is in $\LHV_t$, by Proposition~\ref{LILocSDprop} it is realized by a hidden-variable model $h$ satisfying \SD\ and \LI. Each $\lambda \in \Lambda$ such that $h(\lambda)\da$ induces a function $f_{\lambda} = \prod_i f_{\lambda, i}$, $f_{\lambda, i} : M_i \rarr O_i$, 
whose graph is included in $e$. Since $e$ is realized by $h$, every element of $e$ appears in the graph of $f_{\lambda}$ for some $\lambda$. Hence $e$ is a union of  total strongly deterministic models.
\end{proof}

\noindent Note that the functions $f_{\lambda} = \prod_i f_{\lambda, i}$ are the general form of Mermin-style `instructions'  \cite{mermin1990quantum}, as in the proof of Proposition~\ref{GHZprop}.

In order to simplify notation, we shall consider relational models of arity $n$ of the form $(U, e)$, where $e \subseteq U^n \times U^n$. Thus we use the same underlying set $U$ for both measurements and outcomes at each site. This loses a little generality, but does not change the essentials.

We shall write $\HV(n)$ for the class of models of this form in $\LHV_t$.
We are interested in the  algorithmic problem of determining if a structure $(U, e)$ of arity $n$ is in 
$\HV(n)$.

\begin{proposition}
For each $n$, $\HV(n)$ is in $\NP$.
\end{proposition}
\begin{proof}
It will be convenient to use some notions of finite model theory \cite{libkin2004elements}.

From Proposition~\ref{unionsdprop}, it is clear that $\HV(n)$ is defined by the following second-order formula interpreted over finite structures:
\begin{equation}
\label{secondorderformeq}
\begin{array}{l}
\forall \vec{x}. \exists \vec{y}. R(\vec{x}, \vec{y})  \AND 
[\forall \vec{x}, \vec{y}. \, R(\vec{x}, \vec{y}) \IMP \exists f_1 , \ldots , f_n . \\
\qquad \bigwedge_i f_i(x_i) = y_i 
\AND 
\forall \vec{v}. R(\vec{v}, f_1(v_1), \ldots , f_n(v_n)) ]
\end{array}
\end{equation}
Here we are  interpreting the relation symbol $R$ by the given relation  $e$.

By standard quantifier manipulations, (\ref{secondorderformeq}) can be brought into an equivalent $\Sigma_{1}^{1}$ form, and hence $\HV(n)$ is in $\NP$ \cite{libkin2004elements}.
\end{proof}

\subsection{$\QM$}
We now turn to the question of determining whether a model is realized by a quantum system.

We consider the class of models which can be realized by Hilbert spaces of dimension $d$, which we write as $\QM(d)$. Note that if a model can be realized in dimension $d$, it can be realized in any higher dimension, by working in a  subspace. 

\begin{proposition}
\label{QMdecideprop}
The class $\QM(d)$ is in $\PSPACE$. That is, there is a $\PSPACE$ algorithm to decide, given an empirical model, if it arises from a quantum  system of  dimension $d$.
\end{proposition}
\begin{proof}
We shall give an outline of the construction. The details are tedious but straightforward.

Since we have fixed the dimensions of the Hilbert spaces to be used, the quantum  system can be presented as a list of complex matrices satisfying various properties, each of which can be expressed in terms of the matrix components as equations or inequalities between algebraic expressions. Similarly, the fact that the possibilistic reduction of the system matches the given finite relation can  be expressed in this fashion. We can also use the standard representation of complex numbers as pairs of reals. Hence the existence of a quantum  system realizing the empirical model can be expressed by a first-order sentence $\phi$ over the signature $({+}, 0, {\times}, 1, {<})$, interpreted over the reals. By a classic result of Tarski \cite{tarski1951decision}, the first-order theory of this structure is decidable.

The sentence $\phi$ can in fact be constructed to lie in the \emph{existential fragment}, \ie to have the form  $\exists v_1 \ldots \exists v_k . \psi$, where  $\psi$ is a conjunction of atomic formulas.
This fragment has $\PSPACE$ complexity \cite{canny1988some,renegar1992computational}.
Moreover, the sentence can be constructed in polynomial time from the given relational empirical model.
Hence membership of $\QM(d)$ is in $\PSPACE$.
\end{proof}

The following question seems both interesting and challenging.

\begin{question}
\begin{enumerate}
\item If a model of finite type can be realized by a quantum system, can it always be realized by one of finite dimension?
\item Given a positive answer to (1), is the class $\exists d. \, \QM(d)$ decidable?
Note that, by Proposition~\ref{QMdecideprop}, this class is certainly \emph{computably enumerable}; we can simply run the algorithm for $\QM(d)$ for increasing values of $d$. The question is whether there is an effective bound for this procedure.
\end{enumerate}

\end{question}

\subsection{$\NSp$}
Finally, we consider the class $\NSp$. 

\begin{proposition}
The class $\NSp$ is in $\PSPACE$.
\end{proposition}
\begin{proof}
Following the methods of Proposition~\ref{nsignegprop} and Lemma~\ref{eqproblemm}, we can reduce the question of membership of an empirical model in $\NSp$ to the truth of a sentence in the existential fragment of the first-order theory of the reals, as in the proof of the previous Proposition.
The atomic formulas are the \PNS\ equations, together with strict positivity of the variables, and that they sum to 1. Note that all these expressions are purely additive.
\end{proof}

The nominal complexity suggested by the simple proofs of these results is high; it is likely that better bounds can be achieved.

\begin{question}
Determine the exact complexity of the classes $\LHV$, $\QM$ and $\NSp$.
\end{question}

\subsection*{Acknowledgements}
My thanks to Adam Brandenburger, Noson Yanofsky, Tobias Fritz and Jakub Zavodny for stimulating discussions, and to Philip Atzemoglou, Chris Heunen, Ray Lal and Shane Mansfield for catching a number of  typos. Financial support from 
the EPSRC Senior Research Fellowship EP/E052819/1 and ONR grant 
N000141010357 is gratefully acknowledged.
My thanks also to the two anonymous journal referees for their comments, which led to a number of clarifications and minor corrections.



\end{document}